\documentclass[preprint]{elsarticle}

\usepackage{array,xspace,multirow,hhline,tikz,colortbl,tabularx,booktabs,fixltx2e,amsmath,amssymb,amsfonts,amsthm}
\usetikzlibrary{arrows}
\usepackage{algorithm}
\usepackage{algorithmic}
\usepackage{eqparbox}
\usepackage{nicefrac}
\usepackage{verbatim,ifthen}
\usepackage{enumitem}
\usepackage{pifont}
\usepackage{ifthen}
\usepackage{calrsfs,mathrsfs}
\usepackage{bbding,pifont}
\usepackage{pgflibraryshapes}

\newcommand{\Tube}[6][]%
{   \colorlet{InColor}{#4}
    \colorlet{OutColor}{#5}
    \foreach \I in {1,...,#3}
    {   \pgfmathsetlengthmacro{\h}{(\I-1)/#3*#2}
        \pgfmathsetlengthmacro{\r}{sqrt(pow(#2,2)-pow(\h,2))}
        \pgfmathsetmacro{\c}{(\I-0.5)/#3*100}
        \draw[InColor!\c!OutColor, line width=\r, #1] #6;
    }
}

\usepackage{blkarray}
\newcommand{\matindex}[1]{\mbox{\scriptsize#1}}
	\usepackage{varioref}
		\usepackage{subfigure}

\definecolor{light-gray}{gray}{0.9}

\newtheorem{remark}{Remark}

	\newtheorem{theorem}{Theorem}%
	\newtheorem{example}{Example}
    	
	\newcommand\eat[1]{}

	\usepackage{enumitem}
	\setenumerate[1]{label=\rm(\it{\roman{*}}\rm),ref=({\it\roman{*}}),leftmargin=*}
	\newlength{\wordlength}
	\newcommand{\wordbox}[3][c]{\settowidth{\wordlength}{#3}\makebox[\wordlength][#1]{#2}}

	\newcommand{\set}[1]{\{#1\}}
	\newcommand{\midd}{\mathbin{:}}

	\newcommand{\eqclass}[2][]{\ifthenelse{\equal{#1}{}}{[#2]}{[#2]_{\sim_{#1}}}}

	\newcommand{\ceil}[1]{\lceil #1 \rceil }

	\newcommand{\pref}{\succsim\xspace}
		\newcommand{\spref}{\succ\xspace}

\usepackage{enumitem}
\setenumerate[1]{label=\rm(\it{\roman{*}}\rm),ref=({\it\roman{*}}),leftmargin=*}

\newcommand{\nbh}[1][]{
	\ifthenelse{\equal{#1}{}}{\nu}{\nu(#1)}
}

\newcommand{\cstr}[1][]{
	\ifthenelse{\equal{#1}{}}{\mathscr S}{\cstr(#1)}
}

\newcommand{\choice}[1][]{
	\ifthenelse{\equal{#1}{}}{\mathit{C}}{\choice(#1)}

		\newcommand{\ml}[1][]{\ensuremath{\ifthenelse{\equal{#1}{}}{\mathit{ML}}{\mathit{ML}(#1)}}\xspace}
		\newcommand{\sml}[1][]{\ensuremath{\ifthenelse{\equal{#1}{}}{\mathit{SML}}{\mathit{SML}(#1)}}\xspace}
		\newcommand{\sd}[1][]{\ensuremath{\ifthenelse{\equal{#1}{}}{\mathit{SD}}{\mathit{SD}(#1)}}\xspace}
		\newcommand{\rsd}[1][]{\ensuremath{\ifthenelse{\equal{#1}{}}{\mathit{RSD}}{\mathit{RSD}(#1)}}\xspace}
		\newcommand{\rd}[1][]{\ensuremath{\ifthenelse{\equal{#1}{}}{\mathit{RD}}{\mathit{RD}(#1)}}\xspace}
		\newcommand{\st}[1][]{\ensuremath{\ifthenelse{\equal{#1}{}}{\mathit{ST}}{\mathit{ST}(#1)}}\xspace}
		\newcommand{\bd}[1][]{\ensuremath{\ifthenelse{\equal{#1}{}}{\mathit{BD}}{\mathit{BD}(#1)}}\xspace}
		\newcommand{\pc}[1][]{\ensuremath{\ifthenelse{\equal{#1}{}}{\mathit{PC}}{\mathit{PC}(#1)}}\xspace}
		\newcommand{\dl}[1][]{\ensuremath{\ifthenelse{\equal{#1}{}}{\mathit{DL}}{\mathit{DL}(#1)}}\xspace}
		\newcommand{\ul}[1][]{\ensuremath{\ifthenelse{\equal{#1}{}}{\mathit{UL}}{\mathit{UL}(#1)}}\xspace}

			\newcommand{\indiff}{\ensuremath{\sim}}}
			
			\usepackage{boxedminipage}
			\usepackage{xspace}

    		\renewcommand{\algorithmicrequire}{\wordbox[l]{\textbf{Input}:}{\textbf{Output}:}} 
    		 \renewcommand{\algorithmicensure}{\wordbox[l]{\textbf{Output}:}{\textbf{Output}:}}

\begin{document}

		
				
								
																
																																\title{Simultaneously Achieving Ex-ante and Ex-post Fairness}
		
		\author{Haris Aziz}\ead{haris.aziz@unsw.edu.au}
		
		\address{UNSW Sydney and Data61 CSIRO, Australia}


		    \begin{abstract}
We present a polynomial-time algorithm that computes an ex-ante envy-free lottery over envy-free up to one item (EF1) deterministic allocations. It has the following advantages over a recently proposed algorithm: it does not rely on the linear programming machinery including separation oracles; 
it is SD-efficient (both ex-ante and ex-post); and the ex-ante outcome is equivalent to the outcome returned by the well-known probabilistic serial rule. As a result, we answer a question raised by Freeman, Shah, and Vaish (2020) whether the outcome of the probabilistic serial rule can be implemented by ex-post EF1 allocations. In the light of a couple of impossibility results that we prove, our algorithm can be viewed as satisfying a maximal set of properties. 
Under binary utilities, our algorithm is also ex-ante group-strategyproof and ex-ante Pareto optimal. Finally, we also show that checking whether a given random allocation can be implemented by a lottery over EF1 and Pareto optimal allocations is NP-hard. 
 	\end{abstract}


		\maketitle


			\sloppy
		

\section{Introduction}

Who gets what is a significant and ubiquitous issue. When making any kind of allocation among self-interested agents, fairness is an important concern. Does a fair allocation exist? Is there an efficient algorithm to compute such an allocation? These are important questions that have been studied in fair division for decades. In this paper, we consider the issue of finding probabilistic allocations that are ex-ante and ex-post fair. 

Suppose there are two agents who have additive utilities over three items $a,b,c$. Both agents have the highest value for items $a$, then $b$, and then $c$. From an ex-ante perspective, envy-freeness can be achieved by giving each item to each agent with probability half. However, there are many ways to achieve this expected probability, some perhaps not too fair. For example, the uniform lottery over the following two deterministic allocations: $(\{a,b,c\},\emptyset)$ and $(\emptyset,\{a,b,c\})$.
It may be desirable to achieve both ex-ante envy-freeness and some weaker form of ex-post envy-freeness. For example a uniform lottery over the following allocations is fairer ex-post: $(\{a\},\{b,c\})$ and $(\{b,c\},\{a\})$.

%
%

As seen from the example above, achieving target fairness properties is easy when we consider fractional outcomes or view outcomes from an ex-ante perspective. Implementing such desirable ex-ante outcomes by randomizing over desirable deterministic outcomes can pose interesting challenges (see, e.g. \citep{Aziz18a,BCKM12a}). 
This issue was explored by \citet{FSV20a}. They focussed on ex-ante envy-freeness and ex-post envy-freeness up to one item as the target fairness requirements. Both of the properties are known to be individually achievable. 
An ex-ante envy-free random allocation always exists (for example the outcome of the probabilistic serial rule of \citet{BoMo01a} achieves ex-ante envy-freeness). Similarly, a deterministic envy-free up to one item (EF1) allocation always exists~\citep{Budi11a}. For example, running the round robin sequential algorithm obtains an EF1 allocation~\citep{CKM+16a}. \citet{FSV20a} explore the question of achieving ex-ante envy-freeness and ex-post EF1 \emph{simultaneously}. They showed that there exists a polynomial-time algorithm to compute a lottery over envy-free up to one item allocations that is also ex-ante envy-free.\footnote{\citet{FSV20a} also presented several other results charting the landscape of possibility and impossibility results when considering fairness and efficiency properties ex post and ex-ante. In particular, they study in detail the rule that maximizes ex-ante  Nash welfare. However, they show that the rule cannot be implemented by EF1 allocations. }

The inventive polynomial-time algorithm of \citet{FSV20a} has a couple of possible limitations. Firstly, it requires using the machinery of linear programming separation oracles. It may be desirable to get similar results by simpler combinatorial algorithms. 
Secondly, the algorithm of \citeauthor{FSV20a}  is not ex-post weakly SD (stochastic dominance)-efficient and hence not ex-ante weakly SD-efficient. This is evident from Example 2 of \citeauthor{FSV20a} where they note that their algorithm does not satisfy ordinal efficiency.\footnote{SD-efficiency is also referred to as ordinal efficiency in the literature~\citep{BoMo01a}. } 
The fact that an algorithm is not ex-post weakly SD-efficient implies that it can return a deterministic allocation such that there exists another deterministic allocation that gives each agent strictly more utility for all utility functions consistent with the underlying ordinal preferences. Another implication of violating ex-post weak SD-efficiency
is that all the agents can trade one of their items for another item to get more utility. 
Such unamiguous compromise on welfare can be undesirable. 
For example,  the random serial dictatorship rule (which is ex-post SD-efficient) has received criticism that it is not ex-ante SD-efficient~\citep{BoMo01a}.


We overcome the two limitations discussed above and show that the algorithmic result of \citet{FSV20a} can be achieved in a relatively simpler and faster way while additionally satisfying SD-efficiency. 
To the best of our knowledge, our is the first algorithm to simultaneously satisfy weak SD-efficiency, ex-ante EF, and ex-post EF1. 
The latter two guarantees even hold for all additive utilities consistent with the agents' underlying ordinal preferences. In other words, our algorithm satisfies ex-ante SD-envy-freeness and ex-post SD-EF1. We also show how the algorithm can be further modified by using parametric network flows to additionally achieve both ex-ante and ex-post SD-efficiency. Our results can be viewed as being optimal in the view of the following two impossibility results that we prove. 
Firstly, ex-ante SD-envy-freeness, ex-post EF1, and ex-post Pareto optimality are incompatible. Secondly, ex-ante Pareto optimality and ex-ante SD-envy-freeness  are incompatible.

Our algorithm calls the probabilistic serial algorithm as well as the Birkhoff's decomposition algorithm as subroutines.
\citeauthor{FSV20a} raised the question whether the outcome of the probabilistic serial algorithm can be implemented using ex-post EF1 randomized allocations: ``\emph{we were not able to determine whether the fractional allocation produced by probabilistic
						serial can always be implemented using an ex-post EF1 randomized allocation.}'' We answer the question in the affirmative: our algorithm's outcome is ex-ante equivalent to the outcome of the probabilistic serial rule. In particular, it can be viewed as a desirable way to instantiate the probabilistic serial outcome. 
Under binary utilities, our algorithm is group-strategyproof, ex-ante efficient, ex-ante envy-free, and ex-post EF1. Finally, we also show that checking whether a given random allocation can be represented over a lottery over EF1 and Pareto optimal allocations is NP-hard.

%

\section{Preliminaries}

An allocation problem is a triple $(N,O,u)$ such that $N=\{1,\ldots, n\}$ is the set of agents, $O=\{o_1,\ldots, o_m\}$ is the set of objects, and
$u$ specifies an additive utility function $u_i:O \rightarrow \mathbb{R^+}$. The utility function profile $u$ induces 
the preference profile $\pref=(\pref_1,\ldots, \pref_n)$ which specifies for each agent $i$ his preferences $\pref_i$ over objects in $O$ such that $o \pref_i o'$ if and only if  $u_i(o)\geq u_i(o')$.
We use~$\spref_i$ for the strict part of~$\pref_i$, i.e.,~$o \spref_i o'$ iff~$o \pref_i o'$ but not~$o' \pref_i o$. 
A random allocation $p$ is a $(n\times m)$ matrix $[p_{i,o_j}]$ such that $ p_{i,o_j} \in [0,1]$ for all $i\in N$, and $o_j\in O$; and $\sum_{i\in N}p_{i,o_j}= 1$ for all $o_j\in O$. For a given set $S\subset N$, we will refer by $\pref_S$ the preference profile restricted to agents in $S$.


The value $p_{i,o_j}$ represents the probability of object $o_j$ being allocated to  agent $i$. Each row $p_i=(p_{i,o_1},\ldots, p_{i,o_m})$ represents the allocation of agent $i$. 
The set of columns correspond to the objects $o_1,\ldots, o_m$.
A feasible random allocation is \emph{deterministic} if $p_{i,o}\in \{0,1\}$ for all $i\in N$ and $o\in O$. When we say `an allocation', we will mean random allocation unless we specially specify it is deterministic. 


For any agent $i,j\in N$ and an allocation $p$, the utility of agent $i$ for a bundle $p_j$ is $u_i(p_j)=\sum_{o\in O}p_{j,o}u_i(o)$. 
Given two random allocations $p$ and $q$, $p_i \succsim_i^{SD} q_i$ that is, an agent $i$ \emph{SD~prefers} allocation $p_i$ to allocation $q_i$ if 
$\sum_{o_j\in \set{o_k\midd o_k\succsim_i o}}p_{i,o_j} \ge \sum_{o_j\in \set{o_k\midd o_k\succsim_i o}}q_{i,o_j} \text{ for all } o\in O.$
We write $p_i \succ_i^{SD} q_i$ if $p_i \succsim_i^{SD} q_i$ and not $q_i \succsim_i^{SD} p_i$.	

\paragraph{Fairness Properties}

A random allocation $p$ is \emph{SD-envy-free} if for all $i,j\in N$, $p_i \succsim_i^{SD} p_j$.
An random allocation $p$ is \emph{envy-free (EF)} if 
$u_i(p_i)\geq u_i(p_j)$ for all $i,j\in N$. 
For an agent's allocation $p_j$, we will denote by $p_j^{-o}$ the allocation $p_j$ in which $p_{j,o}$ is set to $0$.
For an agent's allocation $p_j$ and $S\subseteq O$, we will denote by $p_j^{-S}$ the allocation $p_j$ in which $p_{j,o}$ is set to $0$ for all $o\in S$.
A random allocation $p$ is \emph{SD-EF1} if for all $i,j\in N$, 
either $p_i \succsim_i^{SD} p_j$ or $p_i \succsim_i^{SD} p_j^{-o}$ for some
$o$.
$o$. 
A random allocation $p$ is \emph{envy-free up to $k$ items (EFk)} if there exist some $S\subset O$ such that $|S|\leq k$ such that 
$u_i(p_i^{-S})\geq u_i(p_j^{-S})$. 
Note that SD-envy-freeness implies envy-freeness which implies EFk. And SD-EF implies SD-EFk. 

%

\medskip

A given random allocation can be implemented by a lottery over deterministic allocations.\footnote{The statement follows from the well-known Carath{\'{e}}odory's Theorem. } We call the latter an implementation of the given random allocation.
We say that random allocation $p$ satisfies a property $X$ \emph{ex-ante} if the fractional allocation representing $p$ satisfies property $X$. 
When we discuss the ex post properties of a random allocation $p$, we will also need to consider the lottery implementation over deterministic allocations which achieves the random allocation $p$. In that case we say that random assignment with a lottery implementation deterministic allocations over $M_1,\ldots, M_K$ satisfies property $X$ \emph{ex-post} if $M_1,\ldots, M_K$ satisfy property $X$. 
Therefore for any given property for allocations, we consider it ex-ante as well as ex-post. Figure~\ref{fig:fairrelations} shows the key fairness concepts that are appropriate from ex-ante and ex-post perspectives. Note that we do not focus ex-post envy-freeness since a deterministic envy-free allocation is not guaranteed to exist. 
Furthermore, checking whether a deterministic envy-free allocation exists is NP-complete even for 1-0 utilities~\citep{AGMW15a}.

\begin{figure}[h!]
    \begin{center}
        \scalebox{0.85}{
\begin{tikzpicture}
  \tikzstyle{onlytext}=[]
  
  \tikzset{venn circle/.style={circle,minimum width=0mm,fill=#1,opacity=0.1}}

  \node[onlytext] (ex-ante fair) at (0,0) {\begin{tabular}{c}\textbf{ex-ante}\\\textbf{fairness}\end{tabular}};
  
   \draw[-, line width=1pt] (-1,-1) -- (4,-1) ;
  
  \node[onlytext] (ex-ante SD-EF) at (0,-2) {\begin{tabular}{c}{ex-ante}\\{SD-EF}\end{tabular}};
    \node[onlytext] (ex-ante EF) at (0,-4) {\begin{tabular}{c}{ex-ante}\\{EF}\end{tabular}};
	
    \node[onlytext] (ex-post fair) at (3,0) {\begin{tabular}{c}\textbf{ex-post}\\\textbf{fairness}\end{tabular}};
  
    \node[onlytext] (ex-post SD-EF1) at (3,-2) {\begin{tabular}{c}{ex-post}\\{SD-EF1}\end{tabular}};
      \node[onlytext] (ex-post EF1) at (3,-4) {\begin{tabular}{c}{ex-post}\\{EF1}\end{tabular}};

 %
 %
 %
 %

%
%
%
%
%
\draw[->, line width=1pt] (ex-ante SD-EF) -- (ex-ante EF) ;
\draw[->, line width=1pt] (ex-post SD-EF1) -- (ex-post EF1) ;
%
%

\end{tikzpicture}
 }
\end{center}
\caption{\label{fig:fairrelations} Logical relations between fairness concepts.  
}
\end{figure}
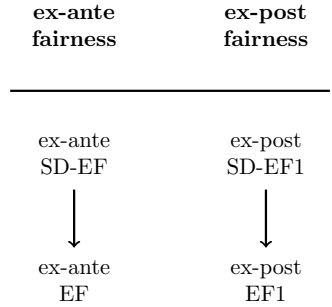

\begin{example}
	Consider the example in which $N=\{1,2\}$, $O=\{a,b,c,d\}$ and 
	the agents have the following utilities over four items.

		\begin{center}
		\setlength{\tabcolsep}{6pt}
		\begin{tabular}{ccccccccc}
			& $a$& $b$& $c$ & $d$\\
			\midrule
			$1$ &$4$&$3$&$2$&$1$\\
			$2$ &$4$&$2$&$3$&$1$\\
		\end{tabular}
	\end{center}

	%
	
	Then, the following is one possible random allocation.
		\begin{center}$p=$
		\begin{blockarray}{ccccccccccc}
			&&\matindex{$a$}&\matindex{$b$}& \matindex{$c$}& \matindex{$d$}&\\
		    \begin{block}{c(cccccccccc)}
				\matindex{$1$}& &$\nicefrac{1}{2}$&$1$&$0$&$\nicefrac{1}{2}$\\
				\matindex{$2$}& &$\nicefrac{1}{2}$&$0$&$1$&$\nicefrac{1}{2}$\\
		    \end{block}
		  \end{blockarray}
		  \end{center}
		  
		  In the allocation, $u_1(p_1)=\frac{1}{2}(4)+ 1(3)+\frac{1}{2}(1)=5.5$ and
		  $u_1(p_2)=\frac{1}{2}(4)+1(2)+\frac{1}{2}(1)=4.5$. Hence agent 1 is not envious of agent 2.
		  
		  
		  Allocation $p$ can be implemented by the following uniform lottery over two deterministic allocations as follows. 
		  \begin{center} $p~~= ~~\frac{1}{2}$
		 	\begin{blockarray}{ccccccccccc}
		 		&&\matindex{$a$}&\matindex{$b$}& \matindex{$c$}& \matindex{$d$}&\\
		 	    \begin{block}{c(cccccccccc)}
		 			\matindex{$1$}& &$1$&$1$&$0$&$0$\\
		 			\matindex{$2$}& &$0$&$0$&$1$&$1$\\
		 	    \end{block}
		 	  \end{blockarray}
		 	 $+~~~\frac{1}{2}$
		   	\begin{blockarray}{ccccccccccc}
		   		&&\matindex{$a$}&\matindex{$b$}& \matindex{$c$}& \matindex{$d$}&\\
		   	    \begin{block}{c(cccccccccc)}
		   			\matindex{$1$}& &$0$&$1$&$0$&$1$\\
		   			\matindex{$2$}& &$1$&$0$&$1$&$0$\\
		   	    \end{block}
		   	  \end{blockarray}
		   	  \end{center}
	
	\end{example}



We say that a deterministic allocation $q$ is \emph{consistent} with a random allocation $p$ if for each $q_{i,o}=1$, we have that  $p_{i,o}>0$.		
For $n=m$, a deterministic allocation can be represented by a permutation matrix in which an entry of one denotes the row agent getting the column object.
A \emph{decomposition} of a random allocation $p$ is a sum $\sum_{i=1}^k \lambda_iP_i$ such that $\lambda_i\in (0,1]$ for $i\in \{1,\ldots,k\}$, $\sum_{i=1}^k\lambda_i=1$, and each $P_i$ is a permutation matrix (consistent with $p$).

\section{The PS-Lottery Algorithm}

In this section, we present our main algorithm that we refer to as the PS-Lottery Algorithm. Before we proceed, 
we summarize two well-known algorithms that we will use as building blocks for our algorithm to simultaneously achieve ex-ante EF and ex-post EF1.

\paragraph{Probabilistic Serial (PS) Algorithm}

The PS rule~\citep{BoMo01a} takes as input the strict ordinal preferences of agents over items as well as the available amounts of each of the items. Agents start eating their most preferred item at unit speed until the item is consumed. They continue eating their most preferred items until all the items are consumed. The outcome is a random allocation in which each agent's probability of getting an item is the fraction of the item that she ate. 
Intially, only presented for the case of single-unit demands, the rule extends seamlessly for the case where agents want to get multiple items~\citep{Koji09a}. 
Although described as a continuous rule where agents eat infinitesimal amounts, the PS outcome can be computed by a discrete algorithm in polynomial time $O(nm)$ (see the appendix).

\paragraph{Birkhoff's Algorithm}

Consider any random allocation with $n$ agents and $n$ items in which each agent gets one unit of items. Birkhoff's algorithm can decompose such a random allocation (which can be  represented by a bistochastic matrix) into a convex combination  of at most $n^2-n+1$ deterministic allocations (represented by permutation matrices)~\citep{Birk46a,LoPl09a}. The following is a description of Birkhoff's algorithm (a formal description is given in the appendix). We initialize $i$ to $1$. For a bistochastic matrix $M$, a permutation matrix $P_i$ consistent with $M$ is guaranteed to exist. Such a permutation matrix corresponds to a perfect matching in a bipartite graph $(N\cup O,E)$ where $(i,o)\in E$ iff $M_{i,o}>0$. Such a perfect matching and hence the permutation matrix can be computed via the Hopcroft-Karp-Karzanov algorithm which takes time $O(n^{2.5})$~\citep{HoKa73a,Karz73a}.
We initialize index $i$ to 1. $M$ is set to $M-\lambda_i{P_i}$ where $\lambda_i\in (0,1]$ is
the smallest non-zero entry in $P_i$. Index $i$ is incremented by one. The updated $M$ is again bistochastic. The process is repeated (say $k-1$ times) until $M$ is the zero matrix. Then $M=\sum_{i=1}^k \lambda_iP_i$. 

\medskip
Now that we have defined the two algorithms, we are in a position to present  Algorithm~\ref{algo:EFEF1}. The high-level description of the algorithm is as follows. 
We first add some dummy items to ensure that there are $nc$ items.
The expanded set of items is called $O'$.
We then simulate PS. We track information about how much of each item has been eaten at time steps $1,\ldots, c$. We use this information to form an allocation $q'$ of items in $O'$ to agents in $N'=\{i_1,\ldots, i_{c}\midd i\in N\}$. 
The agents $i_1,\ldots, i_{c}$ are called the representative agents of each agent $i$. An agent $i_j$'s allocation is what agent $i$ ate in time interval $[j-1,j]$.
Allocation $q'$ can be represented by a bistochastic matrix. We decompose $q'$ into a convex combination of permutation matrices via Birkhoff's algorithm. The permutation matrices are suitably modified to remove the dummy items and also give the allocation of all representatives to the agent they represent. The convex combination over the modified permutation matrices gives us the desired solution, which is both ex-ante EF and ex-post EF1. 

										\begin{algorithm}[h!]
											  \caption{PS-Lottery Algorithm}
											  \label{algo:EFEF1}
			\normalsize
											\begin{algorithmic}
												\REQUIRE  $I=(N,O,\pref)$ where $|N|=n$, $|O|=m$ and $c=\ceil{m/n}$. 							\ENSURE EF fractional allocation $q=\sum_{j=1}^{K}\lambda_jP_i$ where each $P_j$ represents a deterministic EF1 allocation and $K\leq (cn)^2-2cn+2$. 
											\end{algorithmic}
											\begin{algorithmic}[1]
												\normalsize
			 			 \STATE If $m$ is a multiple of $n$, $D=\emptyset$. Else, $D=\{d_1,\ldots, d_{nc-m}\}$. 
			 \STATE $O'\leftarrow O \cup D$ so that $|O'|=cn$.
			 \STATE $N'=\{i_1,\ldots, i_{c}\midd i\in N\}$. 
			 The agents $i_1,\ldots i_c$ are termed as the representatives of agent $i$.
			 \STATE 
		 Set preference profile $\pref'$ of agents in $N'\cup N$ as follows: 	 for all $o,o'\in O$ and for all $i_j$ for $j\in \{1,\ldots, c\}$, $o \pref_{i_j}' o'$ iff $o \pref_i o'$. For all $o\in O$ and $d\in D$, $o \succ_{i_j}' d$. All the ties in $\pref'$ are broken lexicographically. 
			 %
			 \STATE Run PS on instance $(N,O',\pref_N')$ to get a random outcome $r$. \label{step:PS}
			 \STATE For each bundle $r_i$, let agent $i$ re-eat her bundle at unit-speed according to preferences of her representative agents $\pref_{i_k}'$ with each representative agent $i_j$ eating on behalf of agent $i$ in time interval $[j-1,j]$. Let the result of this eating be allocation $q'$ which is an allocation of items $O'$ to agent representatives in $N'$. 
			 \STATE For the (bistochastic) matrix corresponding to $q'$, compute a Birkhoff decomposition 
$q'=\sum_{j=1}^{K}\lambda_jP_j'$ where $K\leq (cn)^2-2cn+2$.
\STATE Convert 	$q'=\sum_{j=1}^{K}\lambda_jP_j'$ into $q=\sum_{j=1}^{K}\lambda_jP_j$ where all the dummy items are ignored and each agent gets the allocation of its representatives. 		 
												\RETURN Allocation $q$ for instance $I$ and its decomposition $\sum_{j=1}^{K}\lambda_jP_j$.
											\end{algorithmic}
										\end{algorithm}

Before we prove the main properties of the PS-Lottery Algorithm, we recall a class of deterministic allocation algorithms.
The \emph{sequential allocation} algorithm takes as input a sequence $\pi$ of turns of the agents and returns a deterministic allocation which is a result of agents picking a most preferred unallocated item in their turn. 
A sequence of turns is called \emph{recursively balanced (RB)} if at each prefix, all agents have the same number of turns, or differ by one.
An RB sequence can  be viewed as agents coming in $c$ rounds. Note that $cn\leq (m+n)$.
In each round except the last one, each agent gets exactly one turn. Since each agent weakly prefers her picked item over all items picked in later rounds, 
it can easily be proved that the outcome of sequential allocation with an RB sequence is EF1~\citep{AHMS19a}.\footnote{In fact an RB allocation satisfies a stronger propery called strong EF1. Stronger EF1 requires that upon removing the same item from agent $i$'s bundle, no other agent $j$ envies $i$, for all $i$ and $j$. 
The property was proposed by \citet{CFS+19}. } Since sequential allocation with an RB sequence only uses ordinal preferences of the agents, it is EF1 with respect to all positive utilities consistent with the ordinal preferences~\citep{AHMS19a} and hence SD-EF. An allocation is called an \emph{RB-allocation} if it is an outcome of sequential allocation with respect to some RB-sequence. We will use the perspective of RB-allocations to establish that our algorithm returns a lottery over EF1 allocations.

\begin{theorem}Let $c=\ceil{m/n}$.
	Algorithm~\ref{algo:EFEF1} is polynomial-time algorithm that takes time $O(({cn})^4)$ that computes a lottery over at most $({cn})^2$ 
deterministic EF1 allocations that is equivalent to the outcome of the probabilistic serial algorithm. 
	\end{theorem}
	\begin{proof}
		Algorithm~\ref{algo:EFEF1} works as follows.
If $m<n$, we set $D=\{d_1,\ldots, d_{n-m}\}$. If $m> n$,  
 we set $D=\{d_1,\ldots,d_{cn -m}\}$. 
 We are now in a position to fix a new allocation instance $I'=(N',O',\pref')$ that only uses ordinal preferences. The item set $O'$ is $O \cup D$ where $|O'|=cn$. The `representative' set $N'$ is $\{i_1,\ldots, i_{c}\midd i\in N\}$. 
Note that the number of representatives $|N'|$ is equal to the number of items $|O'|$. The preferences are consistent with the underlying preference profile. The preferences $\pref'$ of the representatives are set as follows: for all $o,o'\in O$ and for all $i_j$ for $j\in \{1,\ldots, c\}$ $o \pref_{i_j}' o$ iff $o \pref_i o$. For all $o\in O$ and $d\in D$, $o \succ_{i_j}' d$. All the ties in $\pref'$ are broken lexicographically.

 

Note that for the modified allocation problem instance $I'$, an allocation has a corresponding allocation in the original instance $I$: an agent $i$ gets all the allocations of its representatives $i_1,\ldots i_c$. The allocation of dummy items is ignored.

%
Let $q'$ be the allocation as a result of applying PS with agent set $N$ and item set $O'$, but for each $j=0$ to $c-1$, we change the name of each agent $i$ to $i_{j+1}$ in time interval $[j,j+1]$. 
Note that computing $r$ and $q'$ takes time $({cn})^2$.
The allocation has a corresponding bistochastic matrix in which the rows correspond to the representatives and the columns correspond to the items. Each entry in the matrix represents the amount of the corresponding item eaten by the corresponding representative. 

Note that since $q'$ is bistochastic, a permutation matrix $P_k'$ consistent with $q'$ exists by Birkhoff's theorem. 
We want to show that any such matrix $P_k'$ must correspond to an RB-allocation of items in $O'$ to agents in $N$.  
The RB-allocation is viewed as proceeding in rounds. In each round, each of the representatives representing the $n$ agents pick a most preferred available item. In the $j$-th round, the representatives involved are $1_j,\ldots, n_j$. In any $P_k'$, each item is allocated to an agent representative and each agent representative gets one item. In order to establish that $P_k'$ is an RB-allocation of $N$, it is sufficient to prove two claims: (1) no representative agent strictly prefers any item picked in a later round; and (2) within each round, the items allocated to the representative agents are as a result of sequential allocation. 

Claim (1) follows from the fact that no representative $i_j$ strictly prefers any item allocated in a later round. The reason is that when it stopped eating in its turn, it was always eating an item at least as preferred as in later rounds.

Next, we prove Claim (2). Consider any round in which each representative receives one item. We claim that no set of representatives want to reallocate the items given in that round to get an improvement for all representatives in the set. Suppose for contradiction there is a trading cycle in which every agent in the cycle improves: $o_1,1, o_2, 2, \ldots, o_j,j$. 
Representative $1$ prefers item $o_2$ over $o_1$ which means that it started eating $o_1$ after $o_2$ was finished. Since $1$ ate a strictly positive fraction of $o_1$, it implies that $o_1$ finishes strictly after $o_2$. 
By a similar argument each $i\in \{1,\ldots j-1\}$ wants to get $o_{i+1}$ which means that it started eating $o_i$ after $o_{i+1}$ was finished. 
Agent $j$ prefers item $o_1$ over $o_j$ which means that it started eating $o_j$ after $o_1$ was finished which means that $o_j$ finishes strictly after $o_1$.
But then the order of the items according to the finishing times is: $o_1, o_j, o_{j-1},\ldots, o_3, o_2, o_1$. We have shown that $o_1$ has two different finishing times which is a contradiction. Since there exists no trading cycle for representatives in the same round, we know that the items in the round can be allocated as a result of sequential allocation.


From the two claims above, the allocation $P_k'$ is an RB-allocation for agents in $N$ if each agent gets the allocations of its representatives. 
Since any permutation matrix consistent with $q'$ also corresponds to an RB-allocation, we can use $P_k'$ as one of the permutation matrices in which $q'$ is decomposed during Birkhoff's decomposition. We can continue decomposing $q'$ into permutation matrices until we can represent $q'$ by a convex combination of  at most $K\leq ({cn})^2$ permutation matrices $P_1',\ldots, P_K'$. Each permutation matrix in the decomposition can be computed by computing a perfect matching in a corresponding bipartite graph via the Hopcroft-Karp-Karzanov algorithm which takes time $O(({cn})^{2.5})$.

Finally, note that we can convert allocations $(q',P_1',\ldots, P_K')$ for instance $I'$ into the corresponding allocations $(q,P_1,\ldots, P_K)$ for instance $I$. We do so by removing the dummy items and for each $i\in N$, giving the allocations of all the representatives $i_1,\ldots, i_c$ to agent $i$. Note that $q$ is the outcome of running PS on instance $I$. Also, $P_1,\ldots, P_K$ are RB-allocations for instance $I$ and hence EF1 for instance $I$.
		\end{proof}

		\begin{remark}
		Algorithm~\ref{algo:EFEF1} is combinatorial algorithm that computes a lottery over at most $({cn})^2\leq {(m+n)}^2$ deterministic allocations. By Carath{\'{e}}odory's Theorem, any $n\times m$ random allocation that is represented by a convex combination of a given $K$ deterministic allocations, can also be represented by at most $nm+1$ deterministic allocations among the $K$ deterministic allocations. We can reduce the support of the lottery returned by Algorithm~\ref{algo:EFEF1} to one involving at most $nm+1$ deterministic EF1 and SD-efficient allocations as follows. By using Gaussian elimination, we compute the
subset of the set of matrices $\{P_1,\ldots, P_k\}$ that forms the basis of $P_1,\ldots, P_k$. We can then compute a convex combination of the matrices in the basis to achieve the same outcome $q$.		
		%
		%
			\end{remark}
		
We note that whereas our algorithm provides a way to implement PS by EF1 allocations, not every implementation of the PS outcome may satisfy ex-post EF1. For example, consider the case of two agents with identical preferences over two items. In that case, tossing a coin and then giving both items to one agent is ex-ante equivalent to the PS outcome. However, it is not EF1 if agents have strictly positive utilities for both items. 		

Algorithm~\ref{algo:EFEF1}  bears similarities to the exponential-time  Algorithm~1 (Recursive PS) of \citet{FSV20a}. Just like their algorithm, we make agents successively eat one unit of items. Unlike the algorithm of \citeauthor{FSV20a}, we derive the lottery decomposition only after the PS outcome has been computed. In contrast, \citeauthor{FSV20a} probabilistically generate a partial deterministic allocation after each unit time. Their algorithm ``branches out into a polynomial number of subinstances'' a polynomial number of times which makes it an exponential-time algorithm. In order to ensure polynomial-time computability, they resort to a result by \citet{GLS81a} about convex polytopes and separation oracles. \citet{KKN17a} also look at lottery implementation of the PS rule. However, their focus is on allocations with single-unit demand, in which case, any balanced deterministic allocation in the support is trivially EF1. 

Next, we present a simple example showing how our algorithm works. The example has the same preference profile as Example 2 of \citet{FSV20a}.

\begin{example}
	
Consider the example in which $N=\{1,2\}$, $O=\{a,b,c,d\}$ and 
the agents have the following preferences over the four items.
\begin{align*}
	u_1(a)>u_1(b)>u_1(c)>u_1(d)\\
	u_2(a)>u_2(c)>u_2(b)>u_2(d)
	\end{align*}
	
	The exact utilities do not matter since our algorithm only takes into account the underlying ordinal preferences of the agents. 
If we run the PS algorithm, we get the following outcome:
	
	\begin{center}$r=$
	\begin{blockarray}{ccccccccccc}
		&&\matindex{$a$}&\matindex{$b$}& \matindex{$c$}& \matindex{$d$}&\\
	    \begin{block}{c(cccccccccc)}
			\matindex{$1$}& &$\nicefrac{1}{2}$&$1$&$0$&$\nicefrac{1}{2}$\\
			\matindex{$2$}& &$\nicefrac{1}{2}$&$0$&$1$&$\nicefrac{1}{2}$\\
	    \end{block}
	  \end{blockarray}
	  \end{center}
	  
	  Since $m$ is a multiple of $n$, $D=\emptyset$ and hence $O'=O\cup D=O$.
We now show how to achieve our desired lottery to achieve the PS outcome. 
We  run the PS rule on $(N,O',\pref')$ to get allocation $r$. For the outcome, for each agent's bundle, we let successive representative agents to eat exactly one unit of items one by one to get the the following allocation. 

\begin{center}
	$q'=$
\begin{blockarray}{ccccccccccc}
	&&\matindex{$a$}&\matindex{$b$}& \matindex{$c$}& \matindex{$d$}&\\
    \begin{block}{c(cccccccccc)}
		\matindex{$1_1$}& &$\nicefrac{1}{2}$&$\nicefrac{1}{2}$&$0$&$0$\\
		\matindex{$2_1$}& &$\nicefrac{1}{2}$&$0$&$\nicefrac{1}{2}$&$0$\\
		\matindex{$1_2$}& &$0$&$\nicefrac{1}{2}$&$0$&$\nicefrac{1}{2}$\\
		\matindex{$2_2$}& &$0$&$0$&$\nicefrac{1}{2}$&$\nicefrac{1}{2}$\\
    \end{block}
  \end{blockarray}
  \end{center}
  
  The unique decomposition of $q'$ is

 \begin{center}
	 $\frac{1}{2}$
 \begin{blockarray}{ccccccccccc}
 	&&\matindex{$a$}&\matindex{$b$}& \matindex{$c$}& \matindex{$d$}&\\
     \begin{block}{c(cccccccccc)}
 		\matindex{$1_1$}& &$1$&$0$&$0$&$0$\\
 		\matindex{$2_1$}& &$0$&$0$&$1$&$0$\\
 		\matindex{$1_2$}& &$0$&$1$&$0$&$0$\\
 		\matindex{$2_2$}& &$0$&$0$&$0$&$1$\\
     \end{block}
   \end{blockarray}
   $+~~~ \frac{1}{2}$\begin{blockarray}{ccccccccccc}
 	&&\matindex{$a$}&\matindex{$b$}& \matindex{$c$}& \matindex{$d$}&\\
     \begin{block}{c(cccccccccc)}
 		\matindex{$1_1$}& &$0$&$1$&$0$&$0$\\
 		\matindex{$2_1$}& &$1$&$0$&$0$&$0$\\
 		\matindex{$1_2$}& &$0$&$0$&$0$&$1$\\
 		\matindex{$2_2$}& &$0$&$0$&$1$&$0$\\
     \end{block}
   \end{blockarray}
   \end{center}
   
   Translating these for our original instance we get the following decomposition.
  
 \begin{center} $q~~= ~~\frac{1}{2}$
	\begin{blockarray}{ccccccccccc}
		&&\matindex{$a$}&\matindex{$b$}& \matindex{$c$}& \matindex{$d$}&\\
	    \begin{block}{c(cccccccccc)}
			\matindex{$1$}& &$1$&$1$&$0$&$0$\\
			\matindex{$2$}& &$0$&$0$&$1$&$1$\\
	    \end{block}
	  \end{blockarray}
	 $+~~~\frac{1}{2}$
  	\begin{blockarray}{ccccccccccc}
  		&&\matindex{$a$}&\matindex{$b$}& \matindex{$c$}& \matindex{$d$}&\\
  	    \begin{block}{c(cccccccccc)}
  			\matindex{$1$}& &$0$&$1$&$0$&$1$\\
  			\matindex{$2$}& &$1$&$0$&$1$&$0$\\
  	    \end{block}
  	  \end{blockarray}
  	  \end{center}
\end{example}

We also note that if we assume negative utilities instead of positive utilities, we still get SD-EF for PS. For the ex-post guarantee, we get ex-post EF2. 

\section{Additionally Achieving Efficiency}

In this section, we consider the additional issue of efficiency. Before, we proceed, we present some definitions. 

\paragraph{Efficiency Properties}

A random allocation $p$ is \emph{fractional Pareto optimal (fPO)} if there exists no other random allocation $q$ such that $u_i(q_i)\geq u_i(p_i)$ for all $i\in N$ and $u_i(q_i)> u_i(p_i)$ for some $i\in N$.
A deterministic allocation $p$ is \emph{Pareto optimal (PO)} if there exists no other deterministic allocation $q$ such that $u_i(q_i)\geq u_i(p_i)$ for all $i\in N$ and $u_i(q_i)> u_i(p_i)$ for some $i\in N$.
A random allocation $p$ is \emph{SD-efficient} is there exists no random allocation $q$ such that $q_i \succsim_i^{SD} p_i$ for all $i\in N$ and $q_i \succ_i^{SD} p_i$ for some $i\in N$. 
An allocation $p$ is \emph{weakly SD-efficient} is there exists no allocation $q$ such that $q_i \succ_i^{SD} p_i$ for all $i\in N$.
Note that fPO implies PO which implies SD-efficiency which in turn implies weak SD-efficiency. Just as in the case of fairness, we will consider efficiency of both the ex-ante random allocation as well as efficiency properties of the ex-post deterministic allocations that are involved in the lottery.

\medskip

We note that the random allocation maximizing the Nash social welfare is well-known to be equivalent to the competitive equilibrium with equal incomes solution~(see e.g., \citep{Vari74a,DPSV08a}) and satisfies fPO as well as ex-ante envy-freeness. However, due to Theorem 3 of ~\citet{FSV20a}, a rule that is fPO and ex-ante envy-free cannot be 
ex-post EF1.


\medskip

\medskip

Since the outcome returned by Algorithm~\ref{algo:EFEF1} is a lottery implementation of the PS rule outcome, our algorithm also inherits all the desirable ex-ante properties that the PS rule and its outcome are known to satisfy. Note that Algorithm~\ref{algo:EFEF1} first breaks ties in the ordinal preferences before running the PS algorithm. This results in the outcome satisfying weak SD-efficiency rather than SD-efficiency if there are indeed ties in the original preferences. If we care about SD-efficiency, then we do not artificially break any ties and can run the extended probabilistic serial (EPS) algorithm of \citet{KaSe06a}. The EPS algorithm makes coordinated choices for agents to eat one of their most preferred items and uses parametric network flows to compute the outcome. For number of items $m\geq n$, the algorithm takes time $O(m^3\log m)$.\footnote{The orginal EPS algorithm of  \citet{KaSe06a} is presented for the case of single-unit demands. However, it can easily be extended to the case of multiple items (see e.g., the Controlled Cake Eating Algorithm (CCEA) algorithm~\citep{AzYe14a}). CCEA is described in the context of cake cutting with piecewise constant valuations. It also applies to allocation of items: each cake segment can be treated as a separate item. } 

The exact specification of our \emph{EPS-Lottery algorithm} is to take the PS-Lottery algorithm and replace Step~\ref{step:PS} with the following step: Run EPS on instance $(N,O',\pref_N'')$ to get a random outcome $r$. Here, the preference profile $\pref''$ is the same as $\pref'$ except that only ties within $D$ are broken lexicographically and ties are within $O$ are not broken. Therefore the returned outcome $r$ and hence $q'$ is SD-efficient rather than just weak SD-efficient. The argument of implementing the outcome with EF1 deterministic allocations remains unchanged. The running time is unchanged as well as the bottleneck step is to compute a Birkhoff decomposition which takes time $O((cn)^4)$.

Note that if a random allocation $q$ is SD-efficient, then in any decomposition of $q$, each of the deterministic allocations is SD-efficient as well. The reason is that if one of the deterministic allocations is not SD-efficient, then $q$ is not SD-efficient. 
Hence, our algorithm additionally achieves SD-efficiency both ex-ante and ex-post.

\begin{theorem}\label{th:epslottery}
Let $c=\ceil{m/n}$.
The EPS-Lottery Algorithm runs takes time $O(({cn})^4)$  and computes a lottery over at most $({cn})^2\leq (m+n)^2$ deterministic EF1 allocations that is equivalent to the outcome of the extended probabilistic serial algorithm (which is SD-envy-free and SD-efficient). 
	\end{theorem}

We note that our algorithm does not achieve ex-post Pareto optimality. 
One approach to achieving ex-post PO and ex-post EF1 is to check certain random allocations for these properties. Next, we show for an arbitrary random allocation, checking whether it is ex-post EF1 and ex-post Pareto optimal is NP-hard. 

\begin{theorem}
For $n$ agents and $n$ items, checking whether a given random allocation can be implemented by a lottery over EF1 and Pareto optimal allocations is NP-hard.
		For $n$ agents and $n$ items, checking whether a given random allocation can be implemented by a lottery over SD-EF1 and Pareto optimal allocations is NP-hard.
	\end{theorem}
	\begin{proof}
		\citet{AMXY15a} proved that for $n$ agents and $n$ items, checking whether a given random allocation can be implemented by a lottery over balanced Pareto optimal allocations is NP-hard. Their setting assumed ordinal preferences but it works as well for any cardinal preferences consistent with the ordinal preferences. We consider utility functions $u_i$ consistent with ordinal preference $\succ_i$ and assume that $u_i(o)>0$ for all $o\in O$.
Since $u_i(o)>0$ for all $o\in O$, we know that in any unbalanced deterministic allocation one agent $i\in N$ gets zero items and another agent $j$ gets at least two items. Even if one of $j$'s items is removed, $i$ will be envious of $j$. Hence, an unbalanced allocation is not EF1. In the other direction, a balanced allocation gives one item to each agent. Even if an agent $i\in N$ is envious of agent $j$, agent $i$ will not be envious if $j$'s item is removed. We have established that for $n$ agents and $n$ items, the set of deterministic EF1 allocations is equal to the set of deterministic balanced allocations. 
Therefore, the set of deterministic EF1 and Pareto optimal allocations is equivalent to the set of deterministic balanced and Pareto optimal allocations. It follows that checking whether a given random allocation can be implemented by a lottery over EF1 and Pareto optimal allocation is NP-hard.

The same argument also works for the problem of checking whether a given random allocation can be implemented by a lottery over \emph{SD-EF1} and Pareto optimal allocations.
		\end{proof}

\section{Impossibility Results}

We first recall that \citet{FSV20a} proved that even for two agents, ex-ante fPO, ex-ante envy-freeness, and 
ex-post EF1 are incompatible. In this section, we present a couple of more impossibility results. The results are logically incomparable to the main impossibility result of \citet{FSV20a}.
Our first impossibility is the following one. 

\begin{theorem}\label{th:imp}
	Ex-ante SD-EF, ex-post EF1, and ex-post PO are incompatible even for 2 agents.
	\end{theorem}
	\begin{proof}
			Consider the example in which $N=\{1,2\}$, $O=\{a,b_1,b_2,b_3\}$ and 
			the agents have the following utilities over four items.

				\begin{center}
				\setlength{\tabcolsep}{6pt}
				\begin{tabular}{ccccccccc}
					& $a$& $b_1$& $b_2$ & $b_3$\\
					\midrule
					$1$ &$7$&$1$&$1$&$1$\\
					$2$ &$4$&$2$&$2$&$2$\\
				\end{tabular}
			\end{center}
			
			The three items $b_1,b_2,b_3$ are identical items that we refer to as $b$ items. 
		Ex-ante SD-EF implies that each agent in expectation gets $1/2$ of $a$ and $1.5$ units of type $b$ items. 
Our first claim is that in any lottery implementing such an ex-ante SD-EF allocation, there is at least one ex-post allocation in which agent $2$ must get item $a$. This follows from the fact that agent $2$ gets $1/2$ of $a$ in expectation.
		
Our second claim is that in any deterministic 	ex-post EF1 and ex-post PO allocation, agent $2$ cannot get item $a$.	
Suppose for contradiction that agent $2$ gets $a$. Then, EF1 requires that agent $1$ gets at least 2 items of type $b$. But then, agent $1$ can exchange these two items for $a$ to obtain a Pareto improvement. 

From the two claims above, it follows that for the problem instance, there exists no lottery over ex-post EF1 and ex-post PO outcomes that implements the SD-EF random outcome.  
		\end{proof}

	Next, we point out that ex-ante fPO and ex-ante SD-EF are incompatible even for 2 agents.
	The theorem follows directly from Theorem 5 of \citet{AzYe14a} but we re-prove it in our context for the sake of completeness. 
	
	\begin{theorem}\label{cor:imp}
		Ex-ante fPO and ex-ante SD-EF  are incompatible even for 2 agents.
		\end{theorem}
	
	\begin{proof}
	Consider the following two-agent profile.
	
		\begin{center}
		\setlength{\tabcolsep}{6pt}
		\begin{tabular}{ccccccccc}
			& $a$& $b$\\
			\midrule
			$1$ &$u_1(a)$&$u_1(b)$\\
			$2$ &$u_2(a)$&$u_2(b)$\\
		\end{tabular}
	\end{center}

Consider an SD-EF and ex-ante PO allocation $p$.
Suppose $u_1(a), v^{1}_b,  u_2(a), u_2(b) > 0$ in such a way that $u_1(a) > u_1(b)$ and $u_2(a) > u_2(b)$ and $\frac{u_1(a)}{u_1(b)} > \frac{u_2(a)}{u_2(b)}$. 
Due to SD-EF, the outcome should be
		\begin{center}$p=$
		\begin{blockarray}{ccccccccccc}
			&&\matindex{$a$}&\matindex{$b$}\\
		    \begin{block}{c(cccccccccc)}
				\matindex{$1$}& &$\nicefrac{1}{2}$&$\nicefrac{1}{2}$&\\
				\matindex{$2$}& &$\nicefrac{1}{2}$&$\nicefrac{1}{2}$&\\
		    \end{block}
		  \end{blockarray}
		  \end{center}
On the other hand, in order for the mechanism to be ex-ante fPO,  $p_{1,b}=0$ or $p_{2,a}=1$. 
	\end{proof}

	
		\section{Binary Utilities}
		
		We assumed that the agents have additive utilities. If we consider the case in which agents have 1-0 utilities, we can achieve stronger results. We show that our EPS-lottery algorithm satisfies very strong properties when agents have 1-0 utilities. In order to ensure ex-ante efficiency of the EPS-lottery algorithm under 1-0 utilities, we can assume that agents do not consume zero utility items and leave them for the consumption by other agents as is done by the Controlled Cake Eating Algorithm (CCEA) algorithm of \citet{AzYe14a}.
In case this leads to unbalanced allocations, we can make the allocation balanced by adding appropriate number of \emph{extra} dummy items so that we can implement our lottery decomposition algorithm for a balanced allocation.

Before we proceed, let us recall the definition of leximin optimality.
For an allocation $\pi$ we denote by $\vec{u}(\pi) \in \mathbb{R}^n$ the vector of the utilities in $\pi$ sorted in increasing order.
For two vectors $\vec{u}, \vec{v} \in \mathbb{R}^k$, we say that $\vec{u}$ leximin-dominates $\vec{v}$, written $\vec{u} \succ_{lex} \vec{v}$, if there exists an $i \leq k$ such that $\vec{u}_j = \vec{v}_j, \forall j < i$, and $\vec{u}_i > \vec{v}_i$. Finally, $\pi$ is leximin-optimal if there is no $\pi'$ such that $\vec{u}(\pi') \succ_{lex} \vec{u}(\pi)$.

Under 1-0 utilities, it is known that the following rules are equivalent and polynomial-time computable:
(1) leximin rule (2) maximum Nash welfare (MNW) rule (3) competitive equilibirum with equal incomes (CEEI)~\citep{Vari74a} and (4) Controlled Cake Eating Algorithm (CCEA) rule~\citep{AzYe14a} (which can be viewed as an extension for EPS for multi-unit demands that 
is also careful about zero utilities). For example, CEEI and MNW are well-known to be equivalent even for general additive utilities. Under binary utilities, leximin, CEEI, and CCEA are equivalent~\citep{AzYe14a}. 
CCEA satisfies envy-freeness. The conclusion about envy-freeness is also derived from the fact that CEEI outcomes are envy-free (see, e.g. \citet{Vazi07a}).
It is well-known that under additive utilities, the utility profile of the agents is unique (see, e.g., \citet{Vazi07a}). 

For 1-0 utilities, the rules above are known to be ex-ante group-strategyproof (no group of agents can misreport their preferences so that all agents get at least as much utility and at least one agent gets strictly more utility). This fact has been known before as well (see, e.g., \citet{BoMo04a}, \citet{KaSe06a} and \citet{AzYe14a}). Since the rules are equivalent to the leximin rule, the outcome is by definition leximin optimal and hence ex-ante fPO.

We have already shown that an outcome of the EPS rule can be implemented by a lottery over EF1 allocations. Also, every deterministic allocation consistent with the SD-efficient random outcome is SD-efficient (follows from Lemma 2 of \citet{KaSe06a}) and hence ex-post Pareto optimal for binary utilities. Therefore, we achieve ex-post EF1 and ex-post Pareto optimality. 

\begin{theorem}
	For binary utilities, the EPS-Lottery Algorithm is group-strategyproof, ex-ante fPO, ex-post fPO, ex-ante envy-free, and ex-post EF1. Its outcome is ex-ante equivalent to the leximin random allocation as well as the maximum Nash welfare allocation.
	\end{theorem}
	
The theorem above recovers some results that have been proved by \citet{HPPS20a} including their Theorem 4 and Corollary 1.
Our method of achieving a lottery over EF1 and ex-post fPO allocations is different. In their paper, they achieve the EF1 lottery by invoking the bihierarchy framework introduced by \citet{BCKM12a}. 


\section{Conclusion}

We studied the problem of simultaneously achieving desirable fairness properties ex-post and ex-ante. Our main contribution is an algorithm to find a lottery over EF1 allocations that is ex-ante equivalent to the outcome of the (E)PS rule. We noted that we actually compute a lottery over RB-allocations that satisfy strong EF1.  
Figure~\ref{fig:relations} depicts the logical relations between various properties. It also shows some sets of properties that are possible or not possible to satisfy simultaneously.

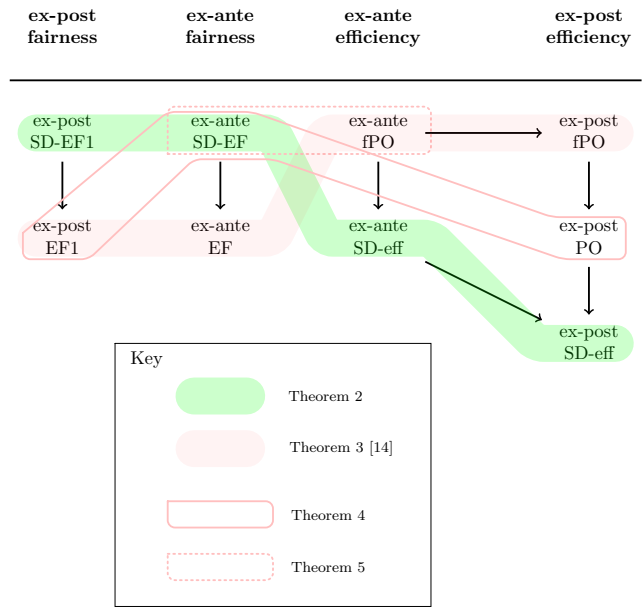
\begin{figure}[h!]
    \begin{center}
        \scalebox{0.7}{
\begin{tikzpicture}
  \tikzstyle{onlytext}=[]

  \tikzset{venn circle/.style={circle,minimum width=0mm,fill=#1,opacity=0.1}}

  \node[onlytext] (ex-ante fair) at (3,0) {\begin{tabular}{c}\textbf{ex-ante}\\\textbf{fairness}\end{tabular}};

  \draw[-, line width=1pt] (-1,-1) -- (11,-1) ;

  \node[onlytext] (ex-ante SD-EF) at (3,-2) {\begin{tabular}{c}{ex-ante}\\{SD-EF}\end{tabular}};
    \node[onlytext] (ex-ante EF) at (3,-4) {\begin{tabular}{c}{ex-ante}\\{EF}\end{tabular}};

    \node[onlytext] (ex-post fair) at (0,0) {\begin{tabular}{c}\textbf{ex-post}\\\textbf{fairness}\end{tabular}};

    \node[onlytext] (ex-post SD-EF1) at (0,-2) {\begin{tabular}{c}{ex-post}\\{SD-EF1}\end{tabular}};
      \node[onlytext] (ex-post EF1) at (0,-4) {\begin{tabular}{c}{ex-post}\\{EF1}\end{tabular}};

 \node[onlytext] (ex-ante fair) at (6,0) {\begin{tabular}{c}\textbf{ex-ante}\\\textbf{efficiency}\end{tabular}};

   \node[onlytext] (ex-ante fPO) at (6,-2) {\begin{tabular}{c}{ex-ante}\\{fPO}\end{tabular}};
   \node[onlytext] (ex-ante SD-eff) at (6,-4) {\begin{tabular}{c}{ex-ante}\\{SD-eff}\end{tabular}};

   \node[onlytext] (ex-post fair) at (10,0) {\begin{tabular}{c}\textbf{ex-post}\\\textbf{efficiency}\end{tabular}};

   \node[onlytext] (ex-post fPO) at (10,-2) {\begin{tabular}{c}{ex-post}\\{fPO}\end{tabular}};
      \node[onlytext] (ex-post PO) at (10,-4) {\begin{tabular}{c}{ex-post}\\{PO}\end{tabular}};
   \node[onlytext] (ex-post SD-eff) at (10,-6) {\begin{tabular}{c}{ex-post}\\{SD-eff}\end{tabular}};

\draw[->, line width=1pt] (ex-ante SD-EF) -- (ex-ante EF) ;
\draw[->, line width=1pt] (ex-post SD-EF1) -- (ex-post EF1) ;
\draw[->, line width=1pt] (ex-ante fPO) -- (ex-ante SD-eff) ;
\draw[->, line width=1pt] (ex-ante fPO) -- (ex-post fPO) ;

\draw[->, line width=1pt] (ex-post fPO) -- (ex-post PO) ;
\draw[->, line width=1pt] (ex-post PO) -- (ex-post SD-eff) ;

\draw[->, line width=1pt] (ex-ante SD-eff) -- (ex-post SD-eff) ;

    \draw [line width=20pt,opacity=0.2,green,line cap=round,rounded corners] (-0.5,-2) -- (4,-2) -- (5,-4) -- (7,-4) -- (9,-6) -- (10.5,-6) ;

 \draw[-,pink, dotted, line width=1pt, rounded corners, line cap=round] (2,-2.4) -- (7,-2.4) -- (7,-1.5) -- (2,-1.5) -- (2,-2.4);

 \draw [line width=20pt,opacity=0.05,color=red,line cap=round,rounded corners] (-0.5,-4) -- (4,-4) -- (5,-2) -- (7,-2) -- (10.5,-2) ;

 \draw[-,pink, solid, line width=1pt, rounded corners, line cap=round] 
 (-0.75,-3.9) -- (-0.75,-4.4) -- (0.5,-4.4) -- (2.5,-2.5) -- (4,-2.5) -- (9.5,-4.4) -- (10.7,-4.4) -- (10.7,-3.6) -- (9.3,-3.6) -- (3.6, -1.6) -- (2, -1.6) -- (-0.75,-3.9);


  \draw [line width=20pt,opacity=0.2,green,line cap=round,rounded corners] (2.5,-7) -- (3.5,-7);
  \node[onlytext] () at (5,-7) {\footnotesize Theorem~\ref{th:epslottery} };
 
  \draw [line width=20pt,opacity=0.05,color=red,line cap=round,rounded corners] (2.5,-8) -- (3.5,-8);
  \node[onlytext] () at (5.3,-8) {\footnotesize Theorem 3~\citep{FSV20a}};

  %
  
  \draw[-,pink, solid, line width=1pt, rounded corners, line cap=round] (2,-9) -- (4,-9) -- (4,-9.5) -- (2,-9.5) -- (2,-9);
     \node[onlytext] () at (5.05,-9.25) {\footnotesize Theorem~\ref{th:imp}};
	 
     \draw[-,pink, dotted , line width=1pt, rounded corners, line cap=round] (2,-10) -- (4,-10) -- (4,-10.5) -- (2,-10.5) -- (2,-10);
        \node[onlytext] () at (5.05,-10.25) {\footnotesize Theorem~\ref{cor:imp}};

  \node[onlytext] () at (1.6,-6.3) {Key};
       \draw[-, line cap=round] (1,-6) -- (7,-6) -- (7,-11) -- (1,-11) -- (1,-6);
  
 
\end{tikzpicture}
 }
\end{center}
\caption{\label{fig:relations} Logical relations between fairness and efficiency concepts.
An arrow from (A) to (B) denotes that (A) implies (B). 
The properties in green are simultaneously satisfied by our algorithm. The combined properties in the pink shapes (dotted, dashed, or shaded) are impossible to simultaneously satisfy. 
}
\end{figure}

We noted that under 1-0 utilities, all meaningful ex-ante and ex-post fairness and efficiency properties are simultaneously satisfied. Coming back to general additive utilities, we recall that our algorithm achieves ex-ante SD-efficiency and ex-ante SD-EF. If we wish to replace ex-ante SD-efficiency with ex-ante fPO, then such an algorithm does not exist in view of Theorem~\ref{cor:imp}. 
Again, note that our algorithm achieves ex-post SD-efficiency, ex-ante SD-EF, and ex-post SD-EF1. Even if we weaken ex-post SD-EF1 to ex-post EF1 but strenghten ex-post SD-efficiency to ex-post PO, we again get an impossibility~(Theorem~\ref{th:imp}).
%


		\section*{Acknowledgements}
		Aziz is supported by a UNSW Scientia Fellowship, and Defence Science and Technology (DST) under the project ``Auctioning for distributed multi vehicle planning'' (DST 9190). He thanks Ethan Brown and Dominik Peters for helpful comments.


%

\bibliographystyle{plainnat}


\newpage

\appendix

\section{Probabilistic Serial (PS) Algorithm}

	We write the formal definition of PS from \citep{Koji09a} as an algorithm. For any $o\in O'\subset O$, let $N(o,O')=\{i\in N\midd o\succ_i b \text{  for every } b\in O'\}$ be the set of agents whose most preferred house in $O'$ is $o$. Let $\max_{N}(O')$ denote $\{o\in O' \midd \exists i\in N \text{ s.t. } o=\max_{\succ_i}(O')\}$.
PS is defined as Algorithm~\ref{algo:subroutine}.

	\begin{algorithm}[htb]
	  \caption{PS}
	  \label{PS}
	\renewcommand{\algorithmicrequire}{\wordbox[l]{\textbf{Input}:}{\textbf{Output}:}}
	 \renewcommand{\algorithmicensure}{\wordbox[l]{\textbf{Output}:}{\textbf{Output}:}}
	\begin{algorithmic}
		\REQUIRE $(N,O,\succ)$
		\ENSURE $p$ the random assignment returned by PS
	\end{algorithmic}
	\algsetup{linenodelimiter=\,}
	  \begin{algorithmic}[1]

	\STATE $j\longleftarrow 0$ ($j$ is the stage of the algorithm)
	\STATE $O^0\longleftarrow O$; $t^0\longleftarrow 0$; $p_{i,o}^0\longleftarrow 0$ for all $i\in N$ and $o\in O$.

	\WHILE{$O^j\neq \emptyset$}
	\FOR{$o\in \max_{N}(O^j)$}
	\STATE $t^{j+1}(o)\longleftarrow \frac{1-\sum_{i\in N} p_{i,o}^j}{|N(o,O^j)|}$
	\COMMENT{finishing times of items that are being eaten}
	\ENDFOR
	\STATE $t^{j+1}\longleftarrow \min_{o\in \max_{N}(O^j)}t^{j+1}(o)$
	\COMMENT{the next time point in which some item is finished}
	\FOR{all $i\in N$ and $o\in O$}
	\IF{$ i\in N(o,O^j)$}
	\STATE  $p_{i,o}^{j+1} \longleftarrow  p_{i,o}^{j}+t^{j+1}-t^j $
	\ELSE
	\STATE $ p_{i,o}^{j+1} \longleftarrow  p_{i,o}^{j} $
	\ENDIF
	
	\ENDFOR
	\STATE $O^{j+1} \longleftarrow O^j\setminus \{o\in \max_{N}(O^j)\midd t^j(o)=t^j\}$
	\COMMENT{Latest eaten items are removed}
	\STATE $j\longleftarrow j+1$
	\ENDWHILE

	\RETURN $p=p^j$

	 \end{algorithmic}
	\label{algo:subroutine}
	\end{algorithm}

\newpage
\section{Birkhoff's Algorithm}

	\begin{algorithm}[htb]
	  \caption{Birkhoff's Algorithm}
	  \label{PS}
	\renewcommand{\algorithmicrequire}{\wordbox[l]{\textbf{Input}:}{\textbf{Output}:}}
	 \renewcommand{\algorithmicensure}{\wordbox[l]{\textbf{Output}:}{\textbf{Output}:}}
	\begin{algorithmic}
		\REQUIRE Bistochastic square matrix $M$ with some non-zero entries and rows $R$ and columns $C$.
		\ENSURE Permutation matrices $P_1,\ldots P_K$ for some $K$ such that $\sum_{i=1}^k \lambda_iP_i=M$. 
	\end{algorithmic}
	\algsetup{linenodelimiter=\,}
	  \begin{algorithmic}[1]

		  \STATE Initalize $i$ to $0$; initalize $M$ to $M'$. 
		  \WHILE{$M'$ has a non-zero entry}
		\STATE $i\longleftarrow i+1$
		  \STATE Compute a permutation matrix $P_i$ consistent with $M'$ as follows. 
		 \COMMENT{One can use the Hopcroft-Karp-Karzanov algorithm to compute $P_i$ as  follows:
Compute a perfect matching $L_i$ in a bipartite graph $(R\cup C,E)$ where $(i,j)\in E$ iff $M_{i,j}'>0$. Let $P_i$ be the  permutation matrix corresponding to matching $L_i$.}
	\STATE $M$ is set to $M-\lambda_i{P_i}$ where $\lambda_i$ is the smallest non-zero entry in $P_i$.
	\ENDWHILE	
\RETURN $P_1,\ldots ,P_i$
\end{algorithmic}
	\label{algo:birkhoff}
	\end{algorithm}

		\end{document}